\documentclass[conference,10pt]{IEEEtran}
\usepackage[left=0.65in, right=0.65in, bottom=1.1in, top=0.75in]{geometry}

\usepackage{bbm}
\usepackage{amsmath}
\usepackage{acronym}  
\usepackage[dvips]{color}
\usepackage{epsf}
\usepackage{times}
\usepackage{epsfig}
\usepackage{notoccite}
\usepackage{graphicx}
\usepackage{epstopdf}
\usepackage{pstricks}
\usepackage{amssymb}
\usepackage{amsthm}
\usepackage{amsxtra}
\usepackage{here}
\usepackage{rawfonts}
\usepackage{float}
\usepackage{times}
\usepackage{url}
\usepackage{cite}
\usepackage{caption}
\usepackage{subcaption}
\usepackage{blindtext}
\usepackage{enumitem}
\usepackage{xcolor,cite,etoolbox}
\usepackage{relsize}
\usepackage{siunitx}
\usepackage{lipsum}
\usepackage{graphicx}
\usepackage{tabularx}
\usepackage{xparse}
\usepackage{array}
\newcolumntype{P}[1]{>{\centering\arraybackslash}p{#1}}
\usepackage{mleftright}

\usepackage{mathtools}

\usepackage{graphics}
\usepackage{physics}
\usepackage{amssymb}
\usepackage{multicol}

\usepackage[nomain,acronym,shortcuts]{glossaries}
\makeglossaries
\newcommand*{\acro}[3][]{\newacronym[#1]{#2}{#2}{#3}}

\acro{D2D}{device-to-device}
\acro{SIR}{signal-to-interference-ratio}
\acro{SINR}{signal-to-interference-plus-noise-ratio}
\acro{PCP}{Poisson cluster process}
\acro{CoMP}{coordinated multi-point}
\acro{BS}{base station} 
\acro{MD-CoMP}{macrodiversity CoMP transmission}
\acro{MAC}{medium-access-control}
\acro{JT-CoMP}{joint transmission CoMP}
\acro{CoMP-JT}{coordinated multipoint joint transmission}
\acro{SBS}{small base station}
\acro{MDSD}{multiple devices to a single device}
\acro{MDS}{maximum distance separable}
\acro{SCN}{small cell network}
\acro{PPP}{Poisson point process}
\acro{TCP}{Thomas cluster process}
\acro{CSI}{channel state information}
\acro{PDF}{probability distribution function}
\acro{PMF}{probability mass function}
\acro{RV}{random variable}
\acro{i.i.d.}{independently and identically distributed}
\acro{MBMS}{multimedia broadcasting multicasting service}
\acro{EE}{energy efficiency}
\acro{HCP}{hard-core placement}
\acro{CCDF}{complementary cumulative distribution function}
\acro{CDF}{cumulative distribution function}
\acro{PC}{probabilistic caching}
\acro{RC}{random caching}
\acro{CPF}{caching popular files} 
\acro{PGFL}{probability generating functional}
\acro{KKT}{Karush-Kuhn-Tucker}
\acro{PGF}{point generating function}
\acro{SCA}{successive convex approximation}
\acro{HD}{high-definition}
\acro{FHD}{full-high-definition}
\acro{UHD}{ultra-high-definition}
\acro{VR}{virtual reality}
\acro{AR}{augmented reality}
\acro{5G}{fifth-generation}
\acro{QoS}{quality-of-service}
\acro{QoE}{quality-of-experience}
\acro{IoT}{Internet of Things}
\acro{MHCPP}{Matern hardcore point process}
\acro{LoS}{line-of-sight}
\acro{NLoS}{non-line-of-sight}
\acro{PSD}{power spectral density}
\acro{MEC}{mobile edge computing}
\acro{E2E}{end-to-end}
\acro{THz}{terahertz}
\acro{CLT}{central limit theorem}
\acro{HQ}{High Quality}
\acro{eMBB}{enhanced mobile broadband}
\acro{URLLC}{ultra reliable low latency communications}
\acro{mmWave}{millimeter wave}
\acro{EVT}{extreme value theory}
\acro{GEV}{generalized extreme value}
\acro{LIS}{large intelligent surface}
\acro{RIS}{reconfigurable intelligent surface}
\acro{RF}{radio frequency}
\acro{UE}{user equipment}
\acro{MIMO}{multiple-input multiple-output}
\acro{EVaR}{entropic value-at-risk}
\acro{DNN}{deep neural network}
\acro{MDP}{Markov decision process}
\acro{RL}{reinforcement learning}
\acro{RNN}{recurrent neural network}
\acro{ANN}{artificial neural networks}
\acro{LSTM}{long short-term memory}
\acro{ReLu}{rectified linear unit}
\acro{VaR}{value-at-risk}
\acro{SNR}{signal-to-noise ratio}
\acro{AoSA}{array of subarray}
\acro{XR}{extended reality}
\acro{AoA}{angle of arrival}
\acro{ULA}{uniform linear array}
\acro{AoD}{angle of departure}
\acro{EM}{electromagnetic}
\acro{HRLLC}{s high-rate and high-reliability low latency communications}
\acro{6DoF}{six degrees of freedom}
\acro{MR}{mixed reality}
\acro{PAPR}{peak to average power ratio}
\acro{OFDM}{orthogonal frequency-division multiplexing}
\acro{OFDMA}{orthogonal frequency-division multiple access}
\acro{SC-FDM}{single carrier frequency-division multiplexing}
\acro{ToA}{time of arrival}
\acro{MUSIC}{multiple signal classification}
\acro{IoE}{Internet of Everything}
\acro{DT}{digital twin}
\acro{PT}{physical twin}
\acro{CT}{cyber twin}
\acro{DRL}{deep reinforcement learning}
\acro{FL}{federated learning}
\acro{DL}{deep learning}
\acro{CRAS}{connected robotics and autonomous system}
\acro{CL}{continual learning}
\acro{MSE}{mean squared error}
\acro{EWC}{elastic weight consolidation}
\acro{ML}{machine learning}
\acro{GD}{gradient descent}
\acro{MLP}{multi layer perceptron}
\acro{TL}{transfer learning}
\acro{AI}{artificial intelligence}
\acro{RoI}{region-of-interest}

\usepackage{datetime}
\usepackage{amssymb}
\usepackage{subcaption}
\usepackage{verbatim}

\newtheorem{proposition}{\bf Proposition}
\newtheorem{lemma}{\bf Lemma}

\IEEEoverridecommandlockouts
\usepackage{cite}
\usepackage{amsmath,amssymb,amsfonts}
\usepackage{algorithmic}
\usepackage{graphicx}
\usepackage{textcomp}
\usepackage{xcolor}
\usepackage{mathtools}
\usepackage{gensymb}
\usepackage{optidef}
\usepackage[linesnumbered,ruled,vlined]{algorithm2e}
\SetKwInput{KwInput}{Input}
\SetKwInput{KwOutput}{Output}
\SetKwInput{KwInitialize}{Initialize}
\usepackage{float}
\begin{document}
\title{Towards a Decentralized Metaverse: Synchronized Orchestration of Digital Twins and Sub-Metaverses
\thanks{This research was supported by the U.S. National Science Foundation under Grants CNS-2210254 and CNS-2007635.}
\vspace{-0.35cm}}
\author{\small Omar Hashash\textsuperscript{\S}, Christina Chaccour\textsuperscript{\S}, Walid Saad\textsuperscript{\S}, Kei Sakaguchi\textsuperscript{\ddag}, and Tao Yu\textsuperscript{\ddag} \\  \textsuperscript{\S}Bradley Department of Electrical and Computer Engineering, Virginia Tech, Arlington, VA, USA.\\ 
\textsuperscript{\ddag}Department of Electrical and Electronic Engineering, Tokyo Institute of Technology, Tokyo, Japan. \\
Emails:\{omarnh, christinac, walids\}@vt.edu, \{sakaguchi, yutao\}@mobile.ee.titech.ac.jp \vspace{-0.4cm}}
\maketitle
\begin{abstract}
Accommodating digital twins (DTs) in the metaverse is essential to achieving \emph{digital reality}. This need for integrating DTs into the metaverse while operating them at the network edge has increased the demand for a \emph{decentralized edge-enabled metaverse}. 
Hence, to consolidate the fusion between real and digital entities, it is necessary to harmonize the interoperability between DTs and the metaverse at the edge. In this paper, a novel decentralized metaverse framework that incorporates DT operations at the wireless edge is presented. In particular, a system of autonomous physical twins (PTs) operating in a massively-sensed zone is replicated as cyber twins (CTs) at the mobile edge computing (MEC) servers. To render the CTs' digital environment, this zone is partitioned and teleported as distributed \emph{sub-metaverses} to the MEC servers.
To guarantee seamless synchronization of the sub-metaverses and their associated CTs with the dynamics of the real world and PTs, respectively, this joint synchronization problem is posed as an optimization problem whose goal is to minimize the average \emph{sub-synchronization time} between the real and digital worlds, while meeting the DT \emph{synchronization intensity} requirements. To solve this problem, a novel iterative algorithm for joint sub-metaverse and DT association at the MEC servers is proposed. This algorithm exploits the rigorous framework of optimal transport theory so as to efficiently distribute the sub-metaverses and DTs, while considering the computing and communication resource allocations.
Simulation results show that the proposed solution can orchestrate the interplay between DTs and sub-metaverses to achieve a $25.75\%$ reduction in the sub-synchronization time in comparison to the signal-to-noise ratio-based association scheme.

\end{abstract}

\vspace{-0.3cm}
\section{Introduction}
\vspace{-0.2cm}
The \emph{metaverse} is perhaps one of the most anticipated technological breakthroughs of the coming decade~\cite{wang2022survey}. In essence, the \emph{metaverse} is a massively scaled and interoperable network of real-time rendered three-dimensional digital worlds. The metaverse will lead to the emergence of a novel suite of hybrid \emph{physical-virtual-digital services} that could revolutionize the interconnection between people, things, and places~\cite{hackl2022navigating,xu2022full}. In essence, building a limitless metaverse of today's real world has various device (e.g. extended reality devices), communication, computing, and \ac{AI} challenges. Chief among those challenges is \emph{the \ac{E2E} synchronization} of the real world with the metaverse and its components such as \acp{DT}. Indeed, achieving this transparent replication imposes a set of stringent wireless network demands such as near-zero \ac{E2E} latency, effective computing, ubiquitous connectivity,  and ultra-high data rates~\cite{chaccour2022seven}. Thus, in an attempt to meet the aforementioned demands, the metaverse must adopt a \emph{decentralized, edge-enabled model}. Instead of limiting the metaverse to a centralized, computationally-draining, and rigid architecture as in the state-of-the-art~\cite{du2022exploring}, this shift enables segmenting the metaverse into a decentralized system of interconnected \emph{sub-metaverses} distributed at the network edge~\cite{dhelim2022edge}. 
Here, a \emph{sub-metaverse} is defined as a digital replica of a \emph{physical space} in the real world.
Furthermore, the decentralization of the metaverse requires decentralizing its key components as well, notably, \acp{DT}. 
\acp{DT} are used to coordinate autonomous \ac{IoE} applications that exist in the real world (e.g., autonomous vehicles) with their digital counterparts in the metaverse~\cite{hashash22}. Thus, it is vital for such real-time \acp{DT} to replicate the functionalities of the underlying physical applications while meeting the \ac{DT} synchronization requirements. Consequently, synchronizing both metaverse and \acp{DT} is pivotal for guaranteeing a high fidelity replica in the digital world.

\indent Recent works in~\cite{lu2021adaptive, han2022dynamic, ismail2022semantic, ng2022stochastic} have studied merging \acp{DT} and the metaverse with \ac{MEC}
in an effort to meet synchronization demands. The authors in~\cite{lu2021adaptive} proposed a deep reinforcement learning approach coupled with transfer learning to solve the \ac{DT}-\ac{MEC} placement and migration problems while minimizing the \ac{DT} synchronization delay.
The work in~\cite{han2022dynamic} studied the resource allocation problem for \ac{IoT} sensing devices with the goal of synchronizing the metaverse by controlling its \acp{DT}' synchronization using a game-theoretic framework. In~\cite{ismail2022semantic}, the authors proposed equipping \ac{IoT} sensing devices with semantic extraction algorithms to minimize the size of the transmitted data in an attempt to achieve enhanced metaverse synchronization.
The work in~\cite{ng2022stochastic} introduced a stochastic semantic resource allocation scheme to enhance the synchronization of virtual transportation networks in the metaverse.
However, these prior works are limited in various ways. First, the work in~\cite{lu2021adaptive} is limited to \ac{DT} synchronization and completely neglects that \acp{DT} are constituents of the metaverse, that in turn should be synchronized as well. Furthermore, the work in~\cite{han2022dynamic} assumes that the metaverse synchronization is achieved through the \ac{DT} synchronization, without strictly differentiating between them as two separate synchronization streams. 
Meanwhile, synchronization of the metaverse in~\cite{ismail2022semantic} and \cite{ng2022stochastic} does not account for the need to  synchronize the metaverse constituents, e.g., \acp{DT}.
Moreover, despite adopting a \ac{MEC} framework in~\cite{han2022dynamic, ismail2022semantic, ng2022stochastic}, these works still consider a centralized metaverse architecture, which cannot accommodate the deployment of \acp{DT} distributed across the edge.
Notably, if done in a centralized way, building a limitless metaverse remains highly intractable. This demands decomposing the metaverse into sub-metaverses that reside at the edge along with the \ac{DT} models.
\emph{Clearly, there is a lack of works that investigate the joint synchronization of \acp{DT} and sub-metaverses in a distributed metaverse framework, while orchestrating the interoperability between them at the edge.}
\indent The main contribution of this paper is a novel decentralized metaverse framework that distributes \acp{DT} and sub-metaverses over the wireless edge network, while preserving upmost synchronization with their physical counterparts in the real world.
In particular, a system of autonomous \acp{PT} that operate in a physical zone are digitally replicated as \acp{CT} at the network's \ac{MEC} servers. To perfectly replicate the physical environment in the digital world, this zone is equipped with massive sensing abilities. Furthermore, the zone is partitioned into separate regions that are teleported as sub-metaverses at the \ac{MEC} servers. To perfectly synchronize the distributed sub-metaverses and their residing \acp{CT} with the corresponding physical counterparts, this joint synchronization problem is posed as an optimization problem whose goal is to minimize the average \emph{sub-synchronization time} with the real world, while satisfying the \ac{DT} \emph{synchronization intensity} requirements. 
To solve this problem, we propose an iterative algorithm based on \emph{optimal transport theory} to determine the joint association of the sub-metaverses and their corresponding \acp{CT} at the \ac{MEC} servers.
We then prove the existence of an optimal solution for the physical region partitioning problem. Subsequently, we find the corresponding map to determine the optimal partitions and \ac{DT} association to minimize the sub-synchronization time, while considering the communications and computing resource allocations in the network. 
\emph{To the best of our knowledge, this is the first work that considers the synchronization of both \acp{DT} and metaverse by addressing the interplay between them at the edge}.
Simulation results show that the proposed solution can provide a tradeoff between associating \acp{DT} and sub-metaverses to achieve a $25.75\%$ reduction in sub-synchronization time compared to a baseline association scheme based on the \ac{SNR}.  
\vspace{-0.2cm}
\section{System Model and Problem Formulation}
\vspace{-0.1cm}
Consider a geographical zone $\mathcal{Z} \subset  \mathbb{R}^{3}$ representing a portion of the real world as shown in Fig.~\ref{System Model}. In this zone, there exists a set $\mathcal{B}$ of $B$ \ac{MEC} servers deployed at the wireless network \acp{BS} to provide digital teleportation, replication, and fine-tuned rendering services. Given that the real world is massive and limitless, it must be digitally represented as a massive system of decentralized sub-metaverses. As such, $B$ \ac{MEC} servers must collectively teleport zone $\mathcal{Z}$ into the digital world as a set of decentralized, yet tightly interconnected sub-metaverses. This guarantees that the computationally intensive teleportation and rendering processes within each sub-metaverse are efficiently deployed and distributed across the network's \ac{MEC} servers. Subsequently, zone $\mathcal{Z}$ is partitioned into a set $\mathcal{A}$ of $A$ non-overlapping regions. Accordingly, each region $a \in \mathcal{A}$ is mapped to a decentralized sub-metaverse $s \in \mathcal{S} \triangleq \{1,2,
\ldots, S\}$, such that each sub-metaverse is associated to a \ac{MEC} server. Thus, we denote the tuple $\boldsymbol{u}_b=(a_b,s_b)$ that maps every region $a \in \mathcal{A}$ to its corresponding sub-metaverse $s \in \mathcal{S}$ at the designated \ac{MEC} server $b \in \mathcal{B}$.\\
\indent Furthermore, there exists a set $\mathcal{K}$ of $K$ of autonomous cyber-physical systems that physically operate in zone $\mathcal{Z}$. To ensure that such systems participate in the metaverse and are autonomously replicated (twinned), a \emph{seamless digital replication} must be performed on each one of them. Thus, each system $k \in \mathcal{K}$ will be operating as a \ac{DT}. Accordingly, for these \acp{DT}, we define the set $\mathcal{P}$ of $P$ \acp{PT}. Such \acp{PT} are served by the $B$ \ac{MEC} servers to simulate and render their corresponding set $\mathcal{C}$ of $C$ \acp{CT}. Thus, each \ac{DT} application $k \in \mathcal{K}$ is represented as a tuple $\boldsymbol{v}_k=(p_k,c_k)$ that maps every \ac{PT} $p \in \mathcal{P}$ to its corresponding \ac{CT} $c \in \mathcal{C}$, where $K = P = C$.
Essentially, the set of \ac{PT} systems that exist within each region $a \in \mathcal{A}$, have corresponding \ac{CT} counterparts that reside in the respective digital replica of this region, i.e., sub-metaverse $s \in \mathcal{S}$.
Consequently, it is necessary that each \ac{MEC} server $b \in \mathcal{B}$ \emph{first} replicates each physical region into the digital world as a sub-metaverse, while also, \emph{simultaneously} guaranteeing the synchronization of the \ac{PT} of each region with its respective \ac{CT} in its sub-metaverse. It is important to note that the \emph{simultaneity} of these two processes is \emph{a necessary condition that guarantees the duality and twinning between the real world and the metaverse.}
\begin{figure}
    \centering
	\includegraphics[scale=0.428]{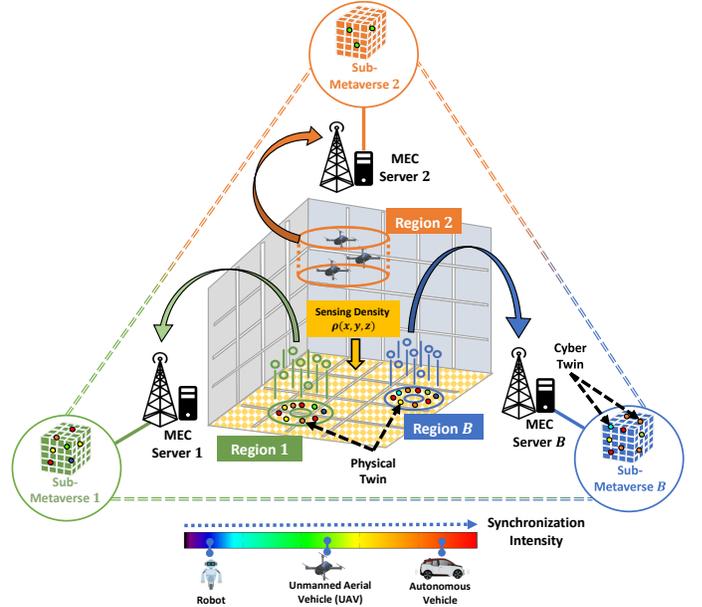}
	\caption{\small{Illustrative figure showcasing the decentralized sub-metaverses and their associated  digital twins distributed at the network edge.}}
	\label{System Model}
	\vspace{-0.45cm}
\end{figure}
\vspace{-0.17cm}
\subsection{Decentralized Metaverse Model}
\vspace{-0.1cm}
To digitally model the physical regions in the metaverse, with high precision, a massive number of sensors must be deployed to replicate aspects of the real world. Because of this large number of sensors deployed, we assume that the inter-spacing between these sensors is relatively minimal. As such, the sensors will be continuously distributed according to a spatial distribution $g(x,y,z)$ that describes the likelihood localization of the sensors around the 3D objects in zone $\mathcal{Z}$, where $x$ is the longitude unit distance, $y$ is the latitude unit distance, and $z$ is the height unit distance. For instance, areas with a large number of physical objects, have a high number of sensors (e.g. highly urban and crowded areas). Meanwhile, areas with a considerable lower physical density, require a smaller number of sensors to be deployed and replicated (e.g. empty fields in rural areas). Also, given that 3D objects in the real world are described with unique attributes of different dimensions, a unit volumetric \emph{sensing density} $\rho(x,y,z)$ ($bps/m^3$) is used to describe the flow of data from these sensors\cite{toumpis2006optimal}. 


To tractably maintain the metaverse synchronized with the real world, we discretize time into independent time slots of duration $\Delta$. During each slot, the sensory data must be uploaded to the network edge. We assume that $\Delta$ is minimal, i.e., $\Delta$ will not affect the precision of the teleported sub-metaverse in mimicking reality. Hence, we consider that each sensor $q \in \mathcal{Q} \triangleq \{1,2,\ldots, Q\}$ forms an independent unit of infinitesimal volume $\epsilon$ centered at $(x,y,z)$ such that the rate in which data is generated from this infinitesimal volume is $\epsilon\rho(x,y,z)$.  Then, the rate at which sensor $q$ uploads its data to \ac{MEC} server $b$ is:
\vspace{-0.3cm}
\begin{equation}
    R_{q,b} (Q_b) =  \frac{W_b^s}{Q_b} \log_2\left(1+ \frac{h_{q,b} \xi_q}{\sigma_b^2}\right), 
\vspace{-0.2cm}
\end{equation}
where $W_b^s$ is the available bandwidth for uploading sensor data to synchronize sub-metaverse $s$ at \ac{MEC} server $b$, $Q_{b}$ is the number of sensors connected to \ac{MEC} server $b$, $h_{q,b}$ is the channel gain between the sensor $q$ and \ac{MEC} server $b$, $\xi_q$ is the transmit power of sensor $q$, and $\sigma_b^2$ is the spectral noise power at \ac{MEC} server $b$.   

Subsequently, the time needed for sensor $q$ to upload its data generated within $\Delta$ to \ac{MEC} server $b$ is: 
\vspace{-0.22cm}
\begin{equation}
    t^{\mathrm{com}}_{q,b} (Q_b) = \frac{\Delta \epsilon}{R_{q,b}(Q_b)}\rho (x,y,z).
    \vspace{-0.15cm}
\end{equation}
Furthermore, the computing time needed to render this infinitesimal volume in sub-metaverse $s \in \mathcal{S}$ at its associated \ac{MEC} server $b \in \mathcal{B}$ can be written as:
\vspace{-0.1cm}
\begin{equation}
    t_{q,b}^{\mathrm{cmp}}(Q_b, \psi_b^s) = \frac{\Lambda \Delta \epsilon}{\psi^{s}_{b}}Q_{b} \rho (x,y,z),
    \vspace{-0.15cm}
\end{equation}
where $\Lambda$ is a coefficient related to the topological complexity of the regions and $\psi_{b}^{s}$ is the number of computing resources assigned for rendering the sub-metaverse $s$ assigned at \ac{MEC} server $b$. 
Then, we define the total delay needed to synchronize the volume unit $\epsilon$ of sensor $q$ from the real world with its representative unit in the designated sub-metaverse at \ac{MEC} server $b$ as:
\vspace{-0.25cm}
\begin{equation} 
t^{\mathrm{sync}}_{q,b}(Q_b, \psi_b^s)= t^{\mathrm{com}}_{q,b}(Q_b) + t^{\mathrm{cmp}}_{q,b}(Q_b, \psi_b^s).
\end{equation}
Moreover, replicating region $a$ that is composed of a large number of these infinitesimal sensors requires synchronizing each sensor with its infinitesimal representation in the sub-metaverse.  
Hence, we define the \emph{sub-synchronization time} as the total delay needed to synchronize region $a$ and its corresponding sub-metaverse $s$ at \ac{MEC} server $b$ as:
\vspace{-0.1cm}
\begin{equation}
    T_b(Q_b, \psi_b^s) = \iiint_{a_{b}} t^{\mathrm{sync}}_{q,b}(Q_b, \psi_b^s) g(x,y,z) \, dx \, dy \, dz.
    \vspace{-0.15cm}
\end{equation}

\subsection{DT Model}
To perfectly consider a synchronized digital world, there is a need to synchronize each of the constituents of the sub-metaverse as well.
In the designated zone, the \acp{PT} are spatially distributed according to a distribution $f(x,y,z)$ in the three dimensional plane. These \acp{PT} are autonomous systems that execute real-time decisions (e.g., autonomous vehicles, drones, robots, etc). Each \ac{DT} $k \in \mathcal{K}$ is described with a \emph{synchronization intensity} $\mu_k$ that captures the maximum allowable time for the \ac{CT} to replicate the same action executed by the \ac{PT}. $\mu_k$ is primarily characterized by the time needed for the PT to execute its task. For instance, an autonomous vehicle will have a high synchronization intensity as it executes its decisions in near real-time (the digital counterpart must execute actions critically). Meanwhile, a robot can tolerate a larger time margin with respect to action execution.\\
\indent Each \ac{PT} is equipped with a large set of sensors that enable a real-time execution of decisions. Concurrently, the corresponding \ac{CT} must to mimic the same action performed by its twinned \ac{PT} while bounded by its intensity limit $\mu_k$. We assume that the inter-region mobility in this system is limited such that each \ac{PT} remains connected to the same \ac{BS}; and thus, each \ac{CT} remains residing in the same sub-metaverse.
Hence, each \ac{PT} $p_k$ must offload its sensory data $\mathcal{D}_k$ collected at each time step to the \ac{MEC} server that hosts its \ac{CT} $c_k$. Subsequently, the achievable uplink rate by $p_k$ and its designated \ac{MEC} server $b$ will be: $r_{k,b} = W_k \log_2 \left(1+ \frac{h_{k,b} \zeta_{k}}{\sigma_b^2}\right),$
where $W_k$ is the uplink bandwidth for \ac{DT} $k$, $h_{k,b}$ is the channel gain between \ac{PT} $p_k$ and \ac{MEC} server $b$, and $\zeta_k$ is the transmit power of $p_k$. 
Thus, the time needed to upload the data generated from $p_k$ to \ac{MEC} server $b$ where the \ac{CT} resides: 
\vspace{-0.18cm}
\begin{equation}
    \tau_{k,b}^{\mathrm{com}} = \frac{D_k}{r_{k,b}}. 
\vspace{-0.15cm}
\end{equation}
In addition, the time needed to execute the corresponding \ac{CT} action at the \ac{MEC} server $b$ is defined as:
\vspace{-0.13cm}
\begin{equation}
 \tau_{k,b}^{\mathrm{cmp}}(\phi_b^k) =  \frac{\Gamma_{k}D_{k}}{\phi_b^{k}},  
 \vspace{-0.17cm}
\end{equation}
where $\Gamma_{k}$ is the twin complexity that relates to the accuracy of the \ac{CT} replication and the intricacy of the underlying physical application, 
and $\phi_{b}^{k}$ is the computing power allocated to the \ac{DT} application $\boldsymbol{v}_k$ at \ac{MEC} server $b$.
Then, the \emph{twinning synchronization time} needed for \ac{DT} $k$ to map the \ac{PT} action in the real world to its \ac{CT} in the sub-metaverse will be:
\begin{equation}
\vspace{-0.2cm}
 \tau_{k,b}^{\mathrm{sync}} (\phi_b^k) =  \tau_{k,b}^{\mathrm{com}} + \tau_{k,b}^{\mathrm{cmp}}(\phi_b^k).  
\end{equation}
\subsection{Problem Formulation}
Our goal is to teleport the zone and the existing \ac{PT} applications into the digital world by distributing them over the network edge. This requires partitioning the regions and mapping them into sub-metaverses at the \ac{MEC} servers, while accommodating the \acp{CT} in their designated sub-metaverses. Thus, our goal to minimize the \emph{sub-synchronization time} between the real world and distributed metaverses, while satisfying the \ac{DT} \emph{synchronization intensity} requirements. In addition to this association, accommodating the sub-metaverses and the \acp{DT} at the edge requires an efficient allocation of the computational resources at the \ac{MEC} server between sub-synchronization and \ac{DT} synchronization. Then, this problem can be formulated as follows:
\vspace{-0.2cm}
\begin{subequations}
\label{opt1}
\begin{IEEEeqnarray}{s,rCl'rCl'rCl}
& \underset{{a_{b, b \in \mathcal{B}}}, \boldsymbol{\psi}, \boldsymbol{\Phi}}{\text{min}} &\quad& \frac{1}{B}\sum_{b \in \mathcal{B}} T_b(Q_b, \psi_b^s),  \label{obj1}\\
&\text{s.t.} && \tau_{k,b}^{\mathrm{sync}}(\phi_b^k) \leq \frac{1}{\mu_{k}}  \quad \forall k \in \mathcal{K}, \forall b \in \mathcal{B}    \label{c2}, \\
& && \psi_{b}^{s} + \sum_{k=1}^{K} \phi_{b}^{k}\leq \Psi_{b} \quad \forall b \in \mathcal{B}  \label{c3}, \\ 
& && \sum_{b\in \mathcal{B}} x_{k,b} \leq 1 \quad  \forall k \in \mathcal{K},\\
& && x_{k,b} \in \{0,1\}, \quad \forall k \in \mathcal{K}, \forall b \in \mathcal{B}\\
& && \bigcup_{b\in \mathcal{B}} a_{b} = \mathcal{Z}, \\ 
& && a_{i} \cap a_{j} =  \emptyset \quad \forall i,j \in \mathcal{A}, i \neq j,
\end{IEEEeqnarray}
\end{subequations}
where $\boldsymbol{\psi} =[\psi_1^{s}, \psi_2^{s},\ldots, \psi_B^{s}]^T $ is the computing resource allocation vector for sub-synchronization. The computing resource allocation and association matrix for \acp{DT} synchronization is given by: \vspace{-0.23cm}
\begin{equation}
\vspace{-0.18cm}
    \boldsymbol{\Phi}=\begin{bmatrix} \phi_1^{1},\ldots, \phi_1^{K}\\  \phi_2^{1},\ldots, \phi_2^{K}\\ \ldots\\ \phi_B^{1}\ldots, \phi_B^{K}\end{bmatrix}.
\end{equation} Here, $x_{k,b}$ is the association variable between the \ac{MEC} server and the \ac{DT}, and $\Psi_b$ is the maximum computing power of \ac{MEC} server $b$.\\
\indent Problem~\eqref{opt1} is challenging to solve as: a) it involves a set of mutually correlated regions, and b) the effective partitioning of the regions depends on the distribution of \acp{DT} and their synchronization requirements. Hence, an effective solution that can address both region partitioning and \ac{DT} association is needed to alleviate the complexity of the problem. Next, we present an iterative algorithm that leverages \emph{optimal transport theory}~\cite{solomon2018optimal} to solve~\eqref{opt1}.
\vspace{-0.25cm}
\section{Optimal Transport Theory for Joint Sub-Metaverse and DTs Association at the Edge}
In this section, we introduce a novel iterative algorithm to solve~\eqref{opt1} by considering an \emph{optimal transport theory} framework~\cite{solomon2018optimal}.
Optimal transport is suitable here because it provides the optimal mapping, from the sensors to the \ac{MEC} servers, which determines the region partitions that yield the minimal sub-synchronization time. 
We first derive the optimal region partitions that minimize the sub-synchronization time based on the metaverse synchronization resources at the \ac{MEC} servers. Then, the \ac{DT}-MEC association is determined under given optimal region partitions and the corresponding computing resources required to synchronize each \ac{DT} are allocated accordingly.\\
\indent First, we consider finding the optimal region partitions that minimize the sub-synchronization time, under a given computing resource allocation for metaverse synchronization. For that, we define $\boldsymbol{\omega} = (x,y,z)$ as the 3D location of sensor $q$, and define $\boldsymbol{\kappa}_b$ as the location of the \ac{MEC} server b.
Then, we reformulate $t^{\mathrm{sync}}_{q,b} = L(g(x,y,z)) F(\boldsymbol{\omega},\boldsymbol{\kappa}_b) $, where $L(g(x,y,z)) = Q\Delta \epsilon \iiint_{a_b} g(x,y,z) \, dx \, dy \, dz$, $F(\boldsymbol{\omega},\boldsymbol{\kappa}_b)= \frac{\rho(\boldsymbol{\omega})}{W_b^s \log_2(1+ \frac{\beta(\boldsymbol{\omega}, \boldsymbol{\kappa}_b)}{\sigma_b^2})} + \frac{\Lambda\rho(\boldsymbol{\omega})}{\psi_b^s}$, and $\beta_{q,b}(\boldsymbol{\omega}, \boldsymbol{\kappa}_b) = h_{q,b} \xi_q$ as the received power from sensor $q$ at \ac{BS} $b$.
Hence, problem~\eqref{opt1} is reduced to a region partitioning optimization problem under given resource allocation, and can be expressed as: 
\begin{subequations}
\label{opt2}
\begin{IEEEeqnarray}{s,rCl'rCl'rCl}
&\underset{a_{b, b \in \mathcal{B}}}{\text{min}} &\quad \frac{1}{B} \sum_{b \in \mathcal{B}} & \iiint_{a_b} L(g(x,y,z)) F(\boldsymbol{\omega},\boldsymbol{\kappa}_b) \label{obj2} \\ 
&&& \times g(x,y,z) \, dx \, dy \, dz , \nonumber \\
&\text{s.t.} & \bigcup_{b\in \mathcal{B}} a_{b} &= \mathcal{Z}, \\
&& a_{i} \cap a_{j}& =  \emptyset \quad \forall i,j \in \mathcal{A}, i \neq j.
\end{IEEEeqnarray}
\end{subequations}
Solving~\eqref{opt2} is challenging since the optimization variables are a set of continuous and dependent region partitions. Thus, to overcome this challenge, we model it as an \emph{optimal transport theory} problem to derive the optimal region partitions that minimize the sub-synchronization time~\cite{mozaffari2017wireless}.
In general, optimal transport is a mathematical framework that considers quantifying the minimal cost for transporting the probability mass of a distribution $\chi_1$ to another one $\chi_2$ on $\Omega \subset \mathbb{R}^w$, by finding an optimal transport map $\Pi$ from $\chi_1$ to $\chi_2$ that minimizes the following function:
\vspace{-0.28cm}
\begin{equation}
\vspace{-0.25cm}
\label{equation}
    \underset{\Pi}{\text{min}} \quad \int_{\Omega} \mathcal{J}(m, n) \chi_1(m) \ dm ; \quad n=\Pi(m),
\end{equation}
where $\mathcal{J}(m, n)$ is the cost of transporting a unit from point $m$ to point $n$. 

Since the sensing density follows a continuous distribution and the \ac{MEC} servers can be considered as discrete points of this distribution, it is natural to model our region partitioning problem as a \emph{semi-discrete optimal transport} problem. Thus, we formulate this mapping problem from the infinitesimal sensor to the \ac{MEC} server as a semi-discrete optimal transport problem that minimizes its sub-synchronization time. In this case, the cost of transportation is considered to be the synchronization time of the infinitesimal volume. Thus, the zone is partitioned into optimal region partitions via this optimal transport map, while noting that $\frac{1}{B}$ is a constant term that does affect the mapping function as it aids in computing the average only. We consider the objective in \eqref{obj2} equivalent to~\eqref{equation} without the constant term. Next, we prove that an optimal solution for this semi-discrete mapping in \eqref{opt2} exists.
\begin{lemma}
 The optimization problem in \eqref{opt2} admits an optimal solution.
\begin{proof}
Let $\alpha_b= \iiint_{a_b} g(x,y,z) \ dx \ dy \ dz$, $ \forall b \in\mathcal{B}$, and define the unit simplex as follows:  
$E= \{ \boldsymbol{\alpha} = (\alpha_1, \alpha_2, \ldots, \alpha_B) \in \mathbb{R}^{B},  \sum_{b=1}^{B} \alpha_b =1, \alpha_b \geq 0 \}$.
Moreover, we define $\chi_{1}(x,y,z) = g(\boldsymbol{\omega})$ and $\mathcal{J}(\boldsymbol{\omega}, \boldsymbol{\kappa}_b) = L(\alpha_b) F(\boldsymbol{\omega},\boldsymbol{\kappa}_b)$.
Clearly, since $F(\boldsymbol{\omega},\boldsymbol{\kappa}_b)$ is continuous, and noting that $L(\alpha_b)$ is differentiable, we can see that $\mathcal{J}(\boldsymbol{\omega}, \boldsymbol{\kappa}_b)$ is continuous. Subsequently, $\mathcal{J}(\boldsymbol{\omega}, \boldsymbol{\kappa}_b)$ is also a lower semi-continuous function. Hence, considering $g$ as a continuous distribution and $\lambda = \sum_{b \in \mathcal{B}} \alpha_b \delta_{\boldsymbol{\kappa}_b}$ as a discrete distribution, with the lower semi-continuous cost function, there theoretically exists an optimal transport map $\Pi$ from $g$ to $\lambda$. Then, for any $\boldsymbol{\alpha} \in E$, problem~\eqref{opt2} admits an optimal solution. Since $E$ is a unit simplex that is a compact set and non-empty, then problem~\eqref{opt2} admits an optimal solution over the entire set $E$. Thus, the lemma is proved.
\end{proof}
\end{lemma}
\vspace{-0.3cm}
\indent Furthermore, we characterize the transport map $\Pi$ which relates each sensor to its corresponding sub-metaverse as follows: 
\vspace{-0.3cm}
\begin{equation}
\label{equation map}
    \biggl\{  \Pi(\boldsymbol{\omega}) = \sum_{b=1}^{B}  \boldsymbol{\kappa}_b    \mathbbm{1}_{a_b}(\boldsymbol{\omega}); \iiint_{a_b} g(\boldsymbol{\omega}) \, d\boldsymbol{\omega} = \alpha_b      \biggl\}
    \vspace{-0.2cm}
\end{equation}
In the following proposition, we set up the solution space for~problem~\eqref{opt2}, which yields the optimal region partitions that minimize the sub-synchronization time.
\begin{proposition}
\label{proposition1}
The optimal region partitioning is given by the following map:
\vspace{-0.25cm}
\small{
\begin{equation}
\label{OptimalMap}
    a_b^{*} = \biggl\{\boldsymbol{\omega}=(x,y,z) |~\alpha_b F(\boldsymbol{\omega},\boldsymbol{\kappa}_b) \leq \alpha_j F(\boldsymbol{\omega},\boldsymbol{\kappa}_j), \forall j\neq b \in \mathcal{B} \biggl\}.
    \vspace{-0.3cm}
\end{equation}
}
\end{proposition}
\begin{proof}
The complete proof is omitted due to space limitations. Essentially, Proposition~\ref{proposition1} can be proved by extending Theorem 1 in\cite{mozaffari2017optimal} to a 3D scenario, while modifying the cost function to include both communication and computing resources at the \ac{BS} and converting it to an uplink scenario.      
\end{proof}
\vspace{-0.3cm}
\begin{algorithm}[t]
\smaller
    \caption{\small Optimal Transport Algorithm for Joint Sub-Metaverse and DT Association}
    \KwInput{$g(x,y,z)$, $\rho(x,y,z)$, $f(x,y,z)$, $Q$, $\mu_k$, $\forall k \in \mathcal{K}$, $\beta_{q,b}$, $\forall b \in \mathcal{B}$, $\forall{q} \in \mathcal{Q}$.}
    \KwOutput{$a_b^*$, $T_b^*$, $\forall b \in \mathcal{B}$, $\boldsymbol{\psi}^*, \boldsymbol{\Phi}^*$.}
    \KwInitialize{Set iterations $\nu = 0$ and $\psi_b^{{s}^{(0)}} =\Psi_b, \forall b \in \mathcal{B}$.}
    Generate initial region partitions $a_b^{(0)}$ and calculate $\alpha_b^{(0)}$.\\
    \Repeat{the convergence of \eqref{equation}}{
    $\nu \leftarrow \nu +1$.\\
    Generate region partitions $a_b^{(\nu)}$ using $\alpha_b^{(\nu-1)}$ and $\psi_b^{s^{(\nu-1)}}$ according to \eqref{OptimalMap}.\\
    Update $\alpha_b^{(\nu)} = \iiint_{a_b^{(\nu)}} g(\boldsymbol{\omega}) \, d\boldsymbol{\omega}$\\
        \For{each region $a_b$}{
             Associate the DTs in region $a_b$ to \ac{MEC} server $b$.\\
             \For{each \ac{DT} $k$ connected to \ac{MEC} server $b$}{
                         Compute $\phi_b^{k^{(\nu)}} = \frac{\Gamma_k}{\frac{1}{\mu_k D_k}- \frac{1}{r_{k,b}}}$.\\
                    }
             Update $\psi_b^{s^{(\nu)}} = \Psi_b - \sum_{k} \phi_b^{k^{(\nu)}}$.\\
        }  
 }
    Determine $a_b^{*} = a_b^{(\nu)}$,  $\boldsymbol{\psi}^{*}= \boldsymbol{\psi}^{(\nu)}$, $\boldsymbol{\Phi}^{*} = \boldsymbol{\Phi}^{(\nu)}$.\\
    Compute the average sub-synchronization time as \eqref{equation} $\times \frac{1}{B}.$
\end{algorithm} 
\setlength{\textfloatsep}{0.1cm}
According to Proposition 1, we can see that $a_b$ and $\alpha_b$ are correlated. Hence, to divide the zone into optimal partitions, it is necessary to first initialize region partitions from an arbitrary partitioning scheme (e.g., Voronoi diagrams). After that, the partitions are iteratively updated based on \eqref{OptimalMap}.
Moreover, after the partitions in each iteration are determined, the \acp{DT} residing in each region are accommodated in the corresponding sub-metaverse, whereby each \ac{DT} should have its synchronization requirement satisfied. This will require providing them with sufficient computing resources at the \ac{MEC} server. Note that the region boundaries impact the number of \acp{PT} within each region, and hence, the \ac{DT} computing power that the \ac{MEC} server can provide for each \ac{DT} application. 
After determining the remaining computing resources at each \ac{MEC} server, this mapping procedure is iterated until convergence.  
Hence, we can reach the optimal association of the regions and \ac{DT} applications that can simultaneously guarantee upmost synchronization of the metaverse and its \ac{DT} components.
This iterative algorithm is summarized in Algorithm 1. \vspace{-0.18cm}
\section{Simulation Results and Analysis}
\vspace{-0.15cm}
For our simulations, we consider a cross-sectional grid zone having dimensions of $2$ km $\times$ $2$ km, in which $B=4$ \acp{BS} are deployed to provide digital teleportation services. The sensing density is modeled as a truncated Gaussian distribution with mean values centered at (750 m, 750 m) and having a standard deviation $\sigma$. This model is suitable to model hotspots of sensors. We consider that the \ac{MEC} servers have a computing power of $\Psi_b = [8, 8, 16, 16]$ \SI{}{GHz} and a bandwidth for metaverse synchronization $W_{b}^s = [10, 10, 20, 20]$ \SI{}{MHz}. We consider three types of \ac{DT} applications having $\mu_k =[10, 100, 1000]$ \SI{}{s^{-1}} and $\zeta_k=[100, 200, 300]$ \SI{}{mW}, respectively. Also, the \acp{PT} are distributed according to a Gaussian mixture model. Unless otherwise stated, we set the following simulation parameters: $Q=25,000$, $\Delta =\SI{1}{ms}$, $\epsilon =\SI{0.01}{m^3}$, $\xi_q =\SI{1}{mW}$, $\sigma_b^2 = \sigma_{o}^2 = \SI{-170}{dBm/Hz}$, $\rho(x,y,z) =\SI{50}{bps/m^3}$, $\Lambda =\SI{5000}{cycles/bit}$, $D_k = \SI{10}{Mb}$, $W_k = 1$ \SI{}{MHz}, and $\Gamma_k =10^4 \SI{}{cycles/bit}$.
As a benchmark, we compare our proposed optimal transport solution to an \ac{SNR}-based association scheme. In the \ac{SNR} scheme, the sensors are associated to the \ac{BS} having the best channel conditions. Accordingly, the \acp{PT} in the resulting region are associated to the same \ac{BS} to execute their twinning processes.     
	\begin{figure}[t!]
		\begin{minipage}{0.49\linewidth}
			\centering
			\includegraphics[scale=0.24]{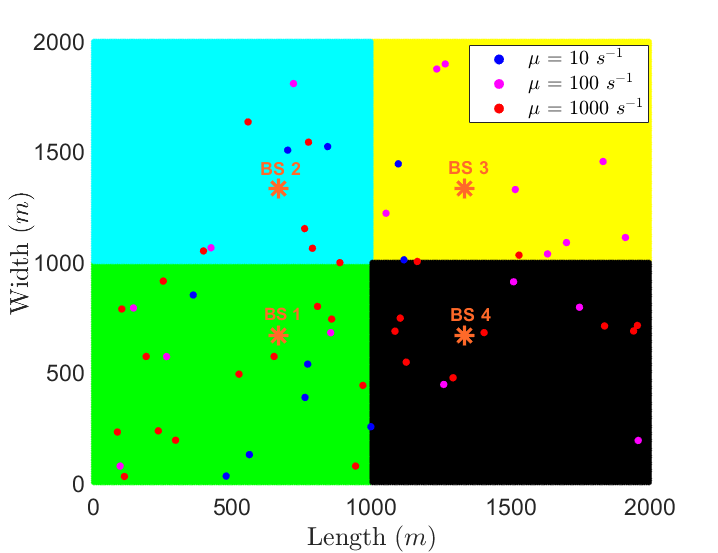}
			\subcaption{}    \label{SNR}
		\end{minipage}
		\begin{minipage}{0.49\linewidth}
			\centering
			\includegraphics[scale=0.24]{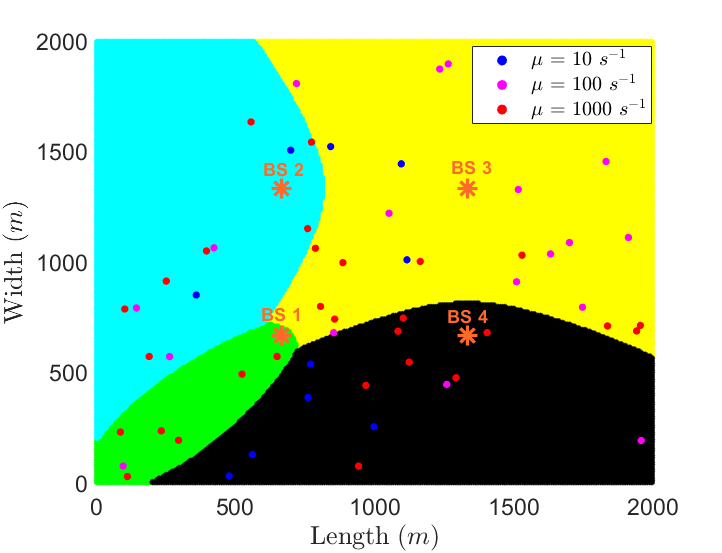}
			\subcaption{}    \label{OT}
		\end{minipage}
		\caption{\small{ Region Partitioning and DT association according to the a) SNR method and b) proposed optimal transport method}.}  \label{partitioning}
		\vspace{-0.1cm}
	\end{figure}


Fig.~\ref{partitioning} compares the partitioning technique adopted by the \ac{SNR}-based approach in Fig.~\ref{SNR}, and our proposed partitioning technique in Fig.~\ref{OT}. Here, the number of \acp{DT} in the network is considered to be $K = 60$. Fig.~\ref{SNR} clearly shows that the regions are equally partitioned based on the \ac{SNR} method. Given that our approach considers the \acp{DT}' operation and its corresponding computing and communication resources, the regions are unequally divided. Moreover, owing to the fact that region 1 has the highest sensing density, it was associated with the smallest region portion and the least number of \acp{DT}. Notably, the other regions characterized with less sensing density are larger in size and encompass a higher number of \acp{DT}. Hence, this figure showcases that our approach provides a \emph{tradeoff between \acp{DT} and sub-metaverses} to guarantee the least sub-synchronization time, unlike the \ac{SNR}-based method.\\
\begin{figure}
    \centering
	\includegraphics[scale=0.35]{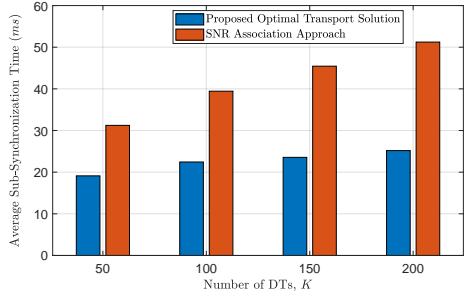}	\caption{\small{The average sub-synchronization time as a function of number of DTs.}}
	\label{Fig3}
	\vspace{-0.4cm}
\end{figure}
\begin{figure}
    \centering
	\includegraphics[scale=0.34]{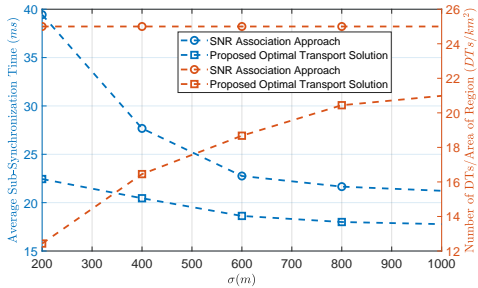}
	\caption{\small{The average sub-synchronization time and the regional density of \acp{DT} versus the standard deviation of the sensing density.}}
	\label{Fig4}
\end{figure}
\indent In Fig.~\ref{Fig3}, the average sub-synchronization time for different number of \acp{DT} in the system is evaluated for the  \ac{SNR}-based association and our proposed optimal transport approach. The proposed method clearly outperforms the \ac{SNR}-based association for all the numbers of \acp{DT}. Fig.~\ref{Fig3} shows that, in our proposed approach, the sub-synchronization time is robust to changes in the number of \acp{DT}. That is, the sub-synchronization time merely increases from $19.11$ \SI{}{ms} to $25.17$ \SI{}{ms} while increasing the number of \acp{DT} from $K=50$ to $K=200$. Meanwhile, when adopting the \ac{SNR} based association, the average sub-synchronization time varies from $31.23$ \SI{}{ms} to $51.23$ \SI{}{ms} for an increase of $K=50$ to $K=200$.\\
\indent Fig.~\ref{Fig4} showcases the average sub-synchronization time and the regional density of \acp{DT} versus the standard deviation of sensing density distribution $\sigma$. First, we can see that the \ac{SNR}-based association does not take into account the number of \acp{DT} nor the available resources. Thus, the regional density of \acp{DT} remains constant despite changes in the sensing density. Meanwhile, in our proposed approach, as $\sigma$ increases, the number of \acp{DT} becomes more uniformly distributed across regions. Thus, the density of \acp{DT} per region increases to asymptotically reach a plateau. Given that a high $\sigma$ characterizes a more uniformly sparsed distribution of sensors, we can see that the regional density reaches a plateau close to that of the \ac{SNR}-based approach. With regards to the sub-synchronization time, we can see that our proposed approach has a $\SI{25.75}{\%}$ lower sub-synchronization time for all values of $\sigma$. Fig.~\ref{Fig4} also shows that the gap between the two methods decreases as $\sigma$ increases. This results from the fact that the number of \acp{DT} becomes more uniformly distributed across regions, which should asymptotically lead to that of the \ac{SNR} method.\vspace{-0.2cm}
\section{Conclusion}
\vspace{-0.1cm}
In this paper, we have proposed a novel framework that decentralizes the metaverse and distributes it along with its \ac{DT} constituents at the edge. 
In particular, we have digitally replicated a physical zone containing autonomous \ac{PT} systems as sub-metaverses that encompass \acp{CT} at the \ac{MEC} servers.
To perform an \ac{E2E} synchronization that utilizes the overall computing and communication resources, we have formulated an optimization problem whose goal is to minimize the sub-synchronization time, while satisfying the \ac{DT} synchronization intensity requirements. This requires partitioning the zone into regions that are associated along with their \acp{PT} as \acp{CT} in their respective sub-metaverses.
To solve this problem, we have presented an iterative algorithm based on optimal transport theory to determine the optimal region mapping and \ac{DT} association that minimizes the sub-synchronization time and synchronizes the \ac{DT} applications. Simulation results show the superiority of our approach compared to \ac{SNR}-based associations that fail to consider the \acp{DT}' operation and its synchronization resources.

\bibliographystyle{IEEEtran}
\def\baselinestretch{0.56}
\bibliography{bibliography}
\end{document}